\documentclass[orivec]{llncs}
\usepackage{amsmath,amssymb,MnSymbol}
\usepackage{hyperref}
\usepackage{authblk} %
\usepackage{color} %
\usepackage{enumitem} %
\usepackage{framed} %
\usepackage[braket,qm]{qcircuit}
\usepackage{tikz}
\usepackage{subfig}
\usepackage{mdframed}
\usepackage{xspace}

\newcommand{\kera}[1]{\ket{#1}\bra{#1}}
\newcommand{\Adv}[2]{\text{\textrm{Adv}}^{#1}_{#2}}
\newcommand{\Game}{\text{\textrm{Game}}}
\renewcommand{\:}{\,:\,}

\newcommand{\bit}{\{0,1\}}

\newcommand{\Meas}{\ensuremath{\mathcal{M}}}
\newcommand{\negl}{\text{\it{negl}}}
\newcommand{\poly}{\text{\it{poly}}}
\newcommand{\from}{\leftarrow}
\newcommand{\PPT}{\textsc{ppt}\xspace}
\newcommand{\QPT}{\textsc{qpt}\xspace}
\newcommand{\RO}{\text{\textrm{RO}}}
\newcommand{\dom}{\text{\it{dom}}}
\newcommand{\IV}{\text{\textrm{IV}}}
\newcommand{\trunc}{\text{\textrm{trunc}}}

\newcommand{\ceil}[1]{\left\lceil#1\right\rceil}
\newcommand{\Id}{\ensuremath{\mathbb{I}}}
\newcommand{\CNOT}{\text{\textrm{CNOT}}}
\newcommand{\Samp}{\text{\textsf{Samp}}}
\newcommand{\veps}{\varepsilon}

\renewcommand{\ts}{(t,\veps)}

\newcommand{\ksp}{\ensuremath{\mathcal{K}}} 
\newcommand{\msg}{\ensuremath{\mathcal{M}}} 
\newcommand{\dig}{\ensuremath{\mathcal{D}}} 
\newcommand{\range}[3]{{#1}_{#2}\|\dots\|{#1}_{#3}}

\newcommand{\ch}{\ensuremath{\mathcal{C}}}
\newcommand{\adv}{\ensuremath{A}}
\newcommand{\game}{\ensuremath{\mathcal{G}}}
\newcommand{\Gint}{\game^{\textrm{int}}}
\newcommand{\Gext}{\game^{\textrm{ext}}}
\newcommand{\trans}{\ensuremath{\mathcal{T}}}
\newcommand{\reduc}{\ensuremath{\mathcal{R}}}
\newcommand{\grayframe}[2]{%
    \begin{mdframed}[style=figstyle,innerleftmargin=10pt,innerrightmargin=10pt]%
    \begin{minipage}{\textwidth}\begin{center}{#1}\end{center}{#2}%
    \end{minipage}\end{mdframed}}
\newcommand{\procedure}[2]{\textsc{#1}(${#2}$)}

\newcommand{\prop}[1]{\text{\textrm{#1}}}
\newcommand{\coll}{\prop{Coll}}
\newcommand{\qcoll}{\prop{CollQ}}
\newcommand{\sqcoll}{\prop{CollSQ}}
\newcommand{\pre}{\prop{Pre}}
\newcommand{\qpre}{\prop{PreQ}}
\newcommand{\sqpre}{\prop{PreSQ}}
\renewcommand{\sec}{\prop{Sec}}
\newcommand{\qsec}{\prop{SecQ}}
\newcommand{\sqsec}{\prop{SecSQ}}
\newcommand{\apre}{\prop{aPre}}
\newcommand{\qapre}{\prop{aPreQ}}
\newcommand{\sqapre}{\prop{aPreSQ}}
\newcommand{\epre}{\prop{ePre}}
\newcommand{\qepre}{\prop{ePreQ}}
\newcommand{\sqepre}{\prop{ePreSQ}}
\newcommand{\asec}{\prop{aSec}}
\newcommand{\qasec}{\prop{aSecQ}}
\newcommand{\sqasec}{\prop{aSecSQ}}
\newcommand{\esec}{\prop{eSec}}
\newcommand{\qesec}{\prop{eSecQ}}
\newcommand{\sqesec}{\prop{eSecSQ}}
\newcommand{\claps}{\prop{CLAPS}}

\newcommand{\prim}[1]{\textsf{#1}}
\newcommand{\ROX}[1]{\prim{ROX}^{(#1)}}
\newcommand{\ROXp}[1]{\overline{\prim{ROX}}^{(#1)}}

\newcommand{\roxpad}{\prim{pad}_{\text{ROX}}}

\newcommand{\imp}{\rightarrow}
\newcommand{\sep}{\nrightarrow}
\newcommand{\timp}{\rightarrowtail}
\newcommand{\tsep}{\nrightarrowtail}
\newcommand{\cimp}{\Rightarrow}
\newcommand{\csep}{\nRightarrow}
\newcommand{\osep}{\nleadsto}

\mdfdefinestyle{figstyle}{ %
  linecolor=black!7, %
  backgroundcolor=black!7, %
  innertopmargin=10pt, %
  innerleftmargin=25pt, %
  innerrightmargin=25pt, %
  innerbottommargin=10pt %
}

\definecolor{White}{rgb}{1,1,1} %
\definecolor{Black}{rgb}{0,0,0} %
\definecolor{LightGray}{rgb}{.8,.8,.8} %
\colorlet{ChannelColor}{LightGray} %
\colorlet{ChannelTextColor}{Black} %
\colorlet{ReadoutColor}{White} %

\title{Quantum security of hash functions and
  property-preservation of iterated hashing}
\author{Ben Hamlin \and Fang Song}
\institute{%
Texas A\&M University \\
\email{\{hamlinb, fang.song\}@tamu.edu}
}
\date{}

\begin{document}

\maketitle

\begin{abstract}

This work contains two major parts: comprehensively studying the
security notions of cryptographic hash functions against quantum
attacks and the relationships between them; and revisiting whether
Merkle-Damg{\aa}rd and related iterated hash constructions preserve
the security properties of the compression function in the quantum
setting. Specifically, we adapt the seven notions in Rogaway and
Shrimpton (FSE'04) to the quantum setting and prove that the seemingly
stronger attack model where an adversary accesses a challenger in
quantum superposition does not make a difference. We confirm the
implications and separations between the seven properties in the
quantum setting, and in addition we construct explicit examples
separating an inherently quantum notion called collapsing from several
proposed properties. Finally, we pin down the properties that are
preserved under several iterated hash schemes. In particular, we prove
that the ROX construction in Andreeva et al. (Asiacrypt'07) preserves
the seven properties in the quantum random oracle model.


\end{abstract}

\section{Introduction}
\label{sec:intro}

\newcommand{\secpar}{\lambda} Cryptographic hash functions, which
produce a short digest on an input message efficiently, are a
ubiquitous building block in modern cryptography. They are
indispensable in constructing key-establishment, authentication,
encryption, digital signature, cryptocurrency, and more, which
constitute the backbone of a secure cyberspace. 
A host of cryptographic hash functions have been
designed~\cite{NISTsha} which have been subject to extensive
cryptanalysis. Most of the constructions follow the \emph{iterated}
hash paradigm, which iterates a compression function on a small
domain.

The emerging technology of quantum computing brings devastating
challenges to cryptography. In addition to breaking widely deployed
public-key cryptography due to Shor's efficient quantum algorithm for
factoring and discrete logarithm, effective quantum attacks on
symmetric primitives have been found in recent years that break
of a variety of message authentication and authenticated encryption
schemes~\cite{KLLNP16,SS16}.

In this work, we revisit two fundamental threads of cryptographic hash
functions in the presence of quantum attacks: modeling basic security
properties and establishing their interrelations; and pinning down
whether the iterated hash constructions \emph{preserve} the security
of the underlying compression functions.

A principal security property is \emph{collision resistance}: It
should be computationally infeasible to compute a \emph{collision}
$(x,x')$ such that $H(x) = H(x')$. Two other basic properties are
preimage resistance (\pre) and second-preimage resistance (\sec).
Rogaway and Shrimpton extend the three and arrive at a total of seven
properties to cope with various scenarios~\cite{RS04}. More
specifically, they consider a family of hash functions $H: \ksp \times
\msg \to \dig$. Conventional $\pre$ and $\sec$ require that under a
\emph{random} key, it is infeasible to find a preimage of a
\emph{random} digest or to find a message that forms a collision
with a given \emph{random} input. They propose two variations named
\emph{always} and \emph{everywhere}. For example, always preimage
resistance (\apre) allows an attacker to pick a key $K$ at will,
and $H_K$ needs to be preimage resistant in the usual sense. This
reflects that real-world hash functions are standalone
(i.e., unkeyed), so it is important to \emph{always} enforce
the property on all members in the hash family. In a complementary
vein \emph{everywhere} preimage resistance (\epre), for instance, asks
about finding a preimage on any digest (i.e., adversarially chosen as
opposed to a random one) being hard. They give a comprehensive
characterization of the seven properties, including both implications
and separations. For instance, they show that while $\coll$ implies
standard $\pre$, there exist $\coll$ hash functions that are not
$\apre$ or $\epre$. This motivates our first question of this work:
\begin{center}
  \emph{How do we model these properties appropriately against quantum
    attacks, and what are the relationships between them?}
\end{center}

Once the appropriate quantum security notions have been nailed down,
we would like to construct hash functions achieving various desired
properties. The dominating design framework is \emph{iterated}
hashing, which takes a compression function on a relatively small
domain and runs it iteratively, with minor variations, to process
longer messages. The Merkle-Damg{\aa}rd
construction~\cite{Merkle89,Dam89} (adopted by SHA-1,2 families) and
the sponge construction~\cite{BDPVA07} (adopted in SHA-3) are notable
examples. As a modular approach to attaining security, researchers ask
whether the iterated hash preserves the security of the compression
function. It is known that Merkle-Damg{\aa}rd is collision resistant
as long as the compression function is collision resistant. However it
does not preserve preimage resistance: There is a preimage-resistant
compression function, such that plugging it into Merkle-Damg{\aa}rd
fails to result in preimage-resistance. Andreeva et al.~\cite{ANPS07}
study several variants of Merkle-Damg{\aa}rd, such as
XOR-linear~\cite{BR97} and Shoup's~\cite{Shoup00} hash schemes, and
determine their security-preserving capabilities. In short, none of
them are able to preserve all seven properties. They therefore propose
a new iterated construction, \emph{ROX}, built on XOR-linear hash, and
prove that it preserves all seven properties in the random oracle
model\footnote{The compression function is not given as a random
oracle. Rather apart from the compression function, the construction
has access to a public random function that is given as a black-box.}.
In contrast, we refer to other constructions as being in the plain
model. We pose the second major question of this work:
\begin{center}
  \emph{Is ROX security preserving in the quantum setting?}
\end{center}

A positive answer will dramatically simplify the design of
secure hash functions to the design of a secure compression function
of a small size. Answering this question, however, could be
challenging and subtle. What we prove classically often fails to carry
over against quantum attacks for some fundamental reasons (e.g.,
no-cloning of quantum states or probabilistic analysis that has no
counterpart in the quantum formalism). There has been extensive work
developing tools for analyzing quantum
security~\cite{Watrous09,Unruh12,Song14,Zhandry12_qprf}. In
particular, Unruh proves that Merkle-Damg{\aa}rd preserves collapsing,
and it can be observed that collision resistance is also preserved in
the quantum setting. More specific to ROX, the random oracle model
faces grave difficulties in the presence of quantum
adversaries~\cite{BDF+11}. For example, classically one can easily
simulate a random oracle by \emph{lazy sampling} the responses upon
every query \emph{on-the-fly}. A quantum query, which can be in
\emph{superposition} of all possible inputs seems to force the
function to be completely specified at the onset. Likewise, the
powerful trick of programming a random oracle, i.e., changing the
outputs on some input points as long as they have not been queried
before, appears impossible if quantum queries are permitted. Recently,
there is progress on restoring proof techniques including programming
a quantum random oracle~\cite{ARU14,Unruh14,ES15,HRS16}.

\paragraph{Our contributions.} We investigate the two questions
systematically in this work. The main results are summarized below.

We formalize the seven security notions in the quantum
setting\footnote{Some standard notions have appeared in the literature
before~\cite{HRS16}.}. Since all properties are described in simple
interactive games, we face two options to modeling quantum attackers
depending on whether the \emph{interface} between the challenger and
the adversary remains classical or can also be quantum. We call the
latter ``fully'' or ``strong'' quantum attacks, reminiscent of an
active line of work recently~\cite{BZ13b,Unruh14,AR17}. This stronger
type of attack is more realistic in some cases than others. Our
interesting finding is that which model we use makes \emph{no}
difference in this setting, by a simple observation of commutativity
of some quantum operators. Namely, the security property (e.g. \apre)
against a quantum adversary and classical communication with the
challenger is equivalent to that where the access to the challenger
can be quantum too.

We depict the landscape of the seven notions in the quantum setting as
well as the collapsing property, by fully determining their
relationships (Figure~\ref{properties-diagram}). For most of the
existing implications and separations in~\cite{RS04}, we apply a
general lifting tool in~\cite{Song14} to make analogous conclusions in
the quantum setting. We construct new examples to separate collapsing
from our quantum notions of $\qasec$ and $\qesec$, and derive other
relations by transitivity. Unruh's separation example between
collapsing and collision resistance~\cite{Unruh16_comm} is the only
one that is relative to an oracle.

We determine the security-preserving capabilities of various iterated
hash constructions. We show that the results in~\cite{ANPS07} (other
than ROX) can be ``lifted'' into a quantum setting. As to ROX, we
adapt techniques of programming a quantum random oracle and show that
ROX preserves all security properties we consider in this work.

\paragraph{Discussion.} As Andreeva et al. remarked in their work, ROX
is proven secure in the random oracle model. Can we design an iterated
hash that is all-preserving in the plain model? Recently there is
another quantum notion extending collision resistance proposed
in~\cite{AMRS18} termed \emph{Bernoulli-preserving}. It implies
collapsing and appears stronger. Do the iterated hash constructions
preserve collapsing and Bernoulli-preserving of the compression
function? Another interesting future direction is to investigate
whether iterated hash can be \emph{amplifying}, especially with the
assistance of a random oracle such as in ROX. Finally, we consider
variants of the Merkle-Damg{\aa}rd and Merkle Tree constructions. Less
is known about the versatile sponge construction in terms of
security-preserving of round functions. It has been shown very
recently that the sponge is collapsing assuming the round functions
are truly random~\cite{CBHSU18}.

\section{Preliminaries}
\label{sec:prelim}

\paragraph{Notations.} Hash-function properties are formulated as
games with a challenger $\ch$ and an adversary $\adv$. $\ch$ and
$\adv$ perform one or more rounds of communication, after which $\ch$
outputs a bit indicating whether $A$ ``won''. Our proofs take the form
of reductions, where winning the game allows us to create an adversary
to win another game that is supposed to be hard. Following
on~\cite{Song14}, we formalize a reduction as a tuple
$(\Gint,\trans,\Gext)$ where $\Gext$ is the game that is assumed to be
hard, $\Gint$ is the game we would like to show to be secure, and
$\trans$ transforms an adversary $\adv$ for $\Gint$ into one for
$\Gext$. If $\trans$ is efficient and maintains $\adv$'s success
probability up to a negligible difference, showing the existence of a
reduction is a proof by contradiction that $\Gint$ is hard.

We are concerned primarily with quantum adversaries. These are
adversaries that run in polynomial time on a quantum computer (\QPT).
We call the probability that this adversary succeeds its
``advantage'', denoted by $\Adv{\prop{prop}}{H}(\adv)$, where $H$ is a
hash function. By $\Adv{\prop{prop}}{H}$, we mean the maximum
advantage over \QPT adversaries. When discussing concrete security, we
say that $H$ is $\ts$-$\prop{prop}$ if for all adversaries $A$ running
in time at most $t$, $\Adv{\prop{prop}}{H}{\adv}\leq\veps$. When the
interaction between $\ch$ and an adversary has two rounds, we
sometimes refer to an adversary as having two parts $(A,B)$. In this
case, they share a state register $S$, which the challenger may not
read or modify. By convention, we use capital letters to indicate
quantum registers. Measuring a quantum register ($\Meas(\cdot)$)
results in a classical value, which we denote with the corresponding
lowercase letter.

We assume there exists a security parameter $n$ for each hash function
that corresponds to the size of a key. A probability is negligible,
denoted $\negl(n)$, if it is less than $\frac{1}{\poly(n)}$, where
$\poly(\cdot)$ is any polynomial function. By $\tau_H$, we mean the
time required to compute $H$. We indicate sampling from a distribution
or receiving a result from a probabilistic algorithm by $x \from S$.
When $S$ is a set, this indicates uniform sampling, unless otherwise
noted.

\paragraph{Quantum random oracles.} One goal of this paper is to
translate results about the ROX construction from the classical (RO)
to the quantum (QRO) random oracle model. In general, results proven
in the classical RO model do not necessarily carry over to a quantum
setting, and even when they do, the techniques often need to be
modified.

Even efficiently simulating a random oracle---a simple task in a
classical setting, since an algorithm can simply lazily answer
$\poly(n)$ queries---is not obviously possible in a quantum setting. A
quantum query could be a superposition of exponentially many inputs,
naively requiring an exponential number of samples from the oracle's
codomain to simulate. Zhandry showed that it is possible to
efficiently simulate a random oracle using $2q$ samples, where $q$ is
the number of queries made to the oracle (Corollary 1 of Theorem 3.1
from~\cite{Zhandry2012}). Whenever we refer to simulating a QRO, we
refer to this technique.

Another property of classical random oracles is that they can be
adaptively programmed. That is, even after a polynomial number of
queries have been made, the algorithm simulating the oracle can change
the output of the oracle at some input points, since it is unlikely
that $\adv$ has seen the output at those points. However, a single
quantum query in superposition can ``see'' the output at all points of
the domain. We use a technique for programming a quantum random oracle
from~\cite{ES15}, which defines a ``witness-search'' game in which an
adversary must guess a ``witness'' $\hat{w}$ with $P(\hat{w})=1$,
given some predicate $P$ and public information $pk$ chosen by the
challenger, given that the challenger knows a witness $w$. The
probability that any \QPT adversary detects adaptive programming at a
point $x$ with $P(\hat{w})=1$ is at most his success probability in
witness search.

\paragraph{Standard hash-function security.} Rogaway and
Shrimpton~\cite{RS04} identify seven properties of hash functions.
These consist of the standard collision resistance (\coll), preimage
resistance (\pre), and second-preimage resistance (\sec), as well as
two stronger variants of each of the latter two---``always'' (\apre,
\asec) and ``everywhere'' (\epre,\esec)---which give the adversary
more power. The following defines standard collision, preimage, and
second-preimage resistance:

\begin{align}
        \Adv{\coll}{H}(A) &= \Pr[x \not= x' \land H_k(x) = H_k(x') \:
                k \from \ksp;
                x,x' \from A(1^n,k)]
        \label{adv-Coll}
\\      \Adv{\pre}{H}(A) &= \Pr[H_k(x) = y \:
                k \from \ksp;
                x' \from M;
                y = H_k(x');
                x \from A(1^n,k,y)]
        \label{adv-Pre}
\\      \Adv{\sec}{H}(A) &= \Pr[x \not= x' \land H_k(x) = H_k(x') \:
                k \from \ksp;
                x' \from M;
                x \from A(1^n,k,x')]
        \label{adv-Sec}
\end{align}

Note that the challenger chooses the key $k$, and in the latter two
properties, challenger chooses the target that the preimage needs to
match. A successful adversary needs to work with non-negligible
probability regardless of what the challenger chooses. One way to
create a stronger property would be to relax this requirement on
either the key or the preimage target.

Allowing the adversary to choose the key results in the ``always''
variants of preimage and second-preimage resistance. Here, the
adversary is given as a pair of algorithms $(A,B)$: $A$ is responsible
for choosing the key, and $B$ is responsible for guessing the
preimage.

\begin{align}
        \Adv{\apre}{H}(A,B) =& \Pr[H_k(x) = y \nonumber
        \\  &\: k,S \from A(1^n);
                x' \from M;
                y = H_k(x');
                x \from B(1^n,S,y)]
        \label{adv-aPre}
\\      \Adv{\asec}{H}(A,B) =& \Pr[x \not= x' \land H_k(x) = H_k(x') \nonumber
        \\ &\: 
                k,S \from A(1^n);
                x' \from M;
                x \from B(1^n,S,x')]
        \label{adv-aSec}
\end{align}

Alternatively, allowing the adversary to choose the target the
preimage must match before knowing the key results in the
``everywhere'' variants of these properties:

\begin{align}
        \Adv{\epre}{H}(A,B) =& \Pr[H_k(x) = y \:
                y,S \from A(1^n);
                k \from \ksp;
                x \from B(1^n,S,k)]
        \label{adv-ePre}
\\      \Adv{\esec}{H}(A,B) =& \Pr[x \not= x' \land H_k(x) = H_k(x') \nonumber
        \\  &\: x',S \from A(1^n);
                k \from \ksp;
                x \from B(1^n,S,k)]
        \label{adv-eSec}
\end{align}

\begin{figure}
\centering

\subfloat[][$\Game_1$]{
\Qcircuit @C=1em @R=.7em {
        & \ustick{k}\cwx[1]  &               &           &                  &
\\      & \multigate{1}{A}   & \ustick{S}\qw & \qw       & \multigate{1}{B} & \qw
\\      & \ghost{A}          & \ustick{X}\qw & \meter    & \ghost{B}        & \qw
\\      & \dstick{y}\cwx     &               &           & \dstick{b}\cwx   &
}
\label{CLAPS-game1}
}
\hspace{1em}
\subfloat[][$\Game_2$]{
\Qcircuit @C=1em @R=.7em {
        & \ustick{k}\cwx[1]  &               &     &     &                  &
\\      & \multigate{1}{A}   & \ustick{S}\qw & \qw & \qw & \multigate{1}{B} & \qw
\\      & \ghost{A}          & \ustick{X}\qw & \qw & \qw & \ghost{B}        & \qw
\\      & \dstick{y}\cwx     &               &     &     & \dstick{b}\cwx   &
}
\label{CLAPS-game2}
}

\caption{Games defining the $\claps$ property.}
\label{CLAPS-games}

\end{figure}
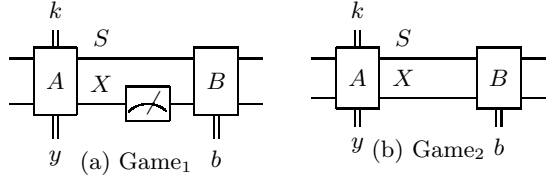

A standard quantum-only property is called
``collapsing''~\cite{Unruh16_comm,Unruh16_coll} ($\claps$). Let $y \in
\dig$ be an element of the digest space of $H_k$. $\claps$ captures
the idea that it should be difficult for an adversary to produce a
``useful'' superposition of elements of the set $H_k^{-1}(y) \subseteq
\msg$. If a hash function is not collapsing, an adversary may be able
to find some input-output pair with desirable properties even if it
can succeed with only negligible advantage in the $\coll$ game.

An adversary for $\claps$ is a pair of \QPT algorithms $(A,B)$. On
input $k$, $A$ outputs quantum registers $S,X$ and a classical
register $y$. We call the adversary ``correct'' if
$\Pr[H_k(\Meas(X))=y] = 1$, and we restrict our attention to correct
adversaries. On input $S,X$, $B$ outputs a classical bit $b$ that
represents a guess whether $X$ has been measured. The collapsing
advantage
$\Adv{\claps}{H}(A,B)=|\Pr[b=1\:\Game{1}]-\Pr[b=1\:\Game{2}]|$, where
$\Game{1,2}$ are as shown in Fig.~\ref{CLAPS-games}.

\section{Quantum security properties of hash functions}
\label{sec:defs}

We adapt the above notions from~\cite{RS04} to a quantum setting
by allowing the adversary to be \QPT, rather than \PPT, as in
the original definitions. The hash function is public, so he can make
superposition queries to it, but all interactions with the challenger
are classical. With the exception of the $\poly(n)$-qubit
state register $S$, we assume that the adversary measures all of its
wires before outputting them. We call these variants \qcoll, \qpre,
etc.

It would be natural to ask whether stronger properties result from
allowing the \emph{interface} between the adversary and the challenger
to be quantum. In other words, the adversary does not measure its
wires before outputting them. At the end, the challenger measures all
registers to determine whether the adversary has succeeded. These
properties, which we call ``strongly quantum'' (SQ), are defined as
follows, where $K$, $Y$, and $X'$ are quantum registers:

\begin{align}
        \Adv{\sqcoll}{H}(A) =& \Pr[x \not= x' \land H_k(x) = H_k(x') \nonumber
        \\      &\:     k \from \ksp;
                        X,X' \from A(1^n,k);
                        x,x' \from \Meas(X,X')]
        \label{adv-Collsq}
\\      \Adv{\sqpre}{H}(A) =& \Pr[H_k(x) = y \nonumber
        \\      &\:     k \from \ksp;
                        x' \from M;
                        y = H_k(x');
                        X \from A(1^n,k,y);
                        x \from \Meas(X)]
        \label{adv-Presq}
\\      \Adv{\sqsec}{H}(A) =& \Pr[x \not= x' \land H_k(x) = H_k(x') \nonumber
        \\      &\:     k \from \ksp;
                        x' \from M;
                        X \from A(1^n,k,x');
                        x \from \Meas(X)]
        \label{adv-Secsq}
\\      \Adv{\sqapre}{H}(A) =& \Pr[H_k(x) = y \nonumber
        \\      &\:     K,S \from A(1^n);
                        x' \from M;
                        Y = U_{H(x')}(K \otimes \ket{0}); \nonumber
        \\      &\ \ \  X \from B(1^n,S,Y);
                        k,x \from \Meas(K,X)]
        \label{adv-aPreSQ}
\\      \Adv{\sqasec}{H}(A) =& \Pr[x \not= x' \land H_k(x) = H_k(x') \nonumber
        \\      &\:     K,S \from A(1^n);
                        x' \from M;
                        X \from B(1^n,S,x');
                        k,x \from \Meas(K,X)]
        \label{adv-aSecSQ}
\\      \Adv{\sqepre}{H}(A) =& \Pr[H_k(x) = y \nonumber
        \\      &\:     Y,S \from A(1^n);
                        k \from \ksp;
                        X \from B(1^n,S,k);
                        x,y \from \Meas(X,Y)]
        \label{adv-ePreSQ}
\\      \Adv{\sqesec}{H}(A) =& \Pr[x \not= x' \land H_k(x) = H_k(x') \nonumber
        \\      &\:     X',S \from A(1^n);
                        k \from \ksp;
                        X \from B(1^n,S,k);
                        x,x' \from \Meas(X,X')]
        \label{adv-eSecSQ}
\end{align}

In (\ref{adv-aPreSQ}), $U_{H(x')}$ is quantum gate that acts as
$U_{H(x')} \: \ket{k}\ket{y} \mapsto \ket{k}\ket{y \oplus H_k(x')}$.
In other words, given a key register $K$ in superposition, it outputs
a superposition of digests for $x'$.

It is easy to see that \sqcoll, \sqpre, and \sqsec (\ref{adv-Collsq},
\ref{adv-Presq}, and \ref{adv-Secsq}) are equivalent to their
counterparts (\ref{adv-Coll}, \ref{adv-Pre}, and \ref{adv-Sec})
defined above: The challenger immediately measures the adversary's
output registers, so without loss of generality, we may assume that
the adversary measures all output registers itself.

As it happens, the other SQ properties
(\ref{adv-aPreSQ}--\ref{adv-eSecSQ}) are equivalent to the above
versions (\ref{adv-aPre}--\ref{adv-eSec}) as well. Intuitively, this
is because, although $A$ can put a superposition of values on its
output register, the challenger never gives this register to $B$. If
the challenger did so, it would be unable to check whether the
adversary had won, since it would no longer have a copy of that
register. Hence, the quantum ``interface'' with the challenger gives
the adversary no additional power in this case.

A more formal proof requires us to show the equivalence of two quantum
circuits. We give the full proof for $\qapre\equiv\sqapre$ in Appendix
\ref{sec:aPre-aPreSQ-equiv}. The proofs for $\qasec$, $\qepre$, and
$\qesec$ are similar, but slightly more straightforward, in that they
do not require Lemma~\ref{aPre-aPreSQ-equiv-lemma2}.

\section{Relations of quantum security properties}
\label{sec:rels}

\begin{figure}
\centering

\subfloat[H][]{
\centering

\resizebox{0.29\textwidth}{!}{
\begin{tikzpicture}[scale=1.5,auto=left]

\node (CLAPS) at (2,4)   {\claps};
\node (Coll)  at (2,3)   {\qcoll};
\node (aSec)  at (1,2.5) {\qasec};
\node (Sec)   at (2,2)   {\qsec};
\node (eSec)  at (3,2.5) {\qesec};
\node (aPre)  at (1,1.5) {\qapre};
\node (Pre)   at (2,1)   {\qpre};
\node (ePre)  at (3,1.5) {\qepre};

\draw[->,>=latex]        (CLAPS) to[bend right] (Coll);
\draw[->,>=latex,dotted] (Coll)  to[bend right] (CLAPS);
\draw[->,>=latex]        (Coll)  to             (Sec);
\draw[->,>=latex]        (Sec)   to             (Pre);
\draw[->,>=latex]        (Coll)  to             (eSec);
\draw[->,>=latex]        (aSec)  to             (aPre);
\draw[->,>=latex]        (aSec)  to             (Sec);
\draw[->,>=latex]        (eSec)  to             (Sec);
\draw[->,>=latex]        (aPre)  to             (Pre);
\draw[->,>=latex]        (ePre)  to             (Pre);

\end{tikzpicture}
}
\label{properties-diagram}

}~\subfloat[H][]{
\centering

\raisebox{2cm}{\resizebox{0.69\textwidth}{!}{
\begin{tabular}{l|cccccccc}
          &  \qsec  & \qasec  & \qesec  &  \qpre  & \qapre  & \qepre  & \qcoll  & \claps
\\ \hline
   \qsec  &    -    & $\csep$ & $\csep$ & $\cimp$ & $\csep$ & $\csep$ & $\csep$ & $\tsep$
\\ \qasec & $\cimp$ &    -    & $\csep$ & $\cimp$ & $\cimp$ & $\csep$ & $\csep$ & $\tsep$
\\ \qesec & $\cimp$ & $\csep$ &    -    & $\cimp$ & $\csep$ & $\csep$ & $\csep$ & $\tsep$
\\ \qpre  & $\csep$ & $\csep$ & $\csep$ &    -    & $\csep$ & $\csep$ & $\csep$ & $\tsep$
\\ \qapre & $\csep$ & $\csep$ & $\csep$ & $\cimp$ &    -    & $\csep$ & $\csep$ & $\tsep$
\\ \qepre & $\csep$ & $\csep$ & $\csep$ & $\cimp$ & $\csep$ &    -    & $\csep$ & $\tsep$
\\ \qcoll & $\cimp$ & $\csep$ & $\cimp$ & $\cimp$ & $\csep$ & $\csep$ &    -    & $\osep$
\\ \claps & $\timp$ & $\sep$  & $\timp$ & $\timp$ & $\sep$  & $\sep$  & $\imp$  &    -
\end{tabular}
}}
\label{properties-table}

}

\caption{
The relationships between the properties defined in
Sect.~\ref{sec:defs}. In Fig.~\ref{properties-diagram}, solid arrows
indicate implications. Everywhere an arrow (or a transitive
implication) is absent indicates a separation. In
Fig.~\ref{properties-table}, $\imp$ and $\sep$ are implications with
explicit proofs, and $\timp$ and $\tsep$ hold by transitivity, $\cimp$
and $\csep$ indicate quantumly ``lifted'' classical reductions. The
dotted arrow in Fig.~\ref{properties-diagram} and $\osep$ in
Fig.~\ref{properties-table} indicate that only a relativized
separation has been shown~\cite{Unruh16_comm}.
}
\label{properties}

\end{figure}
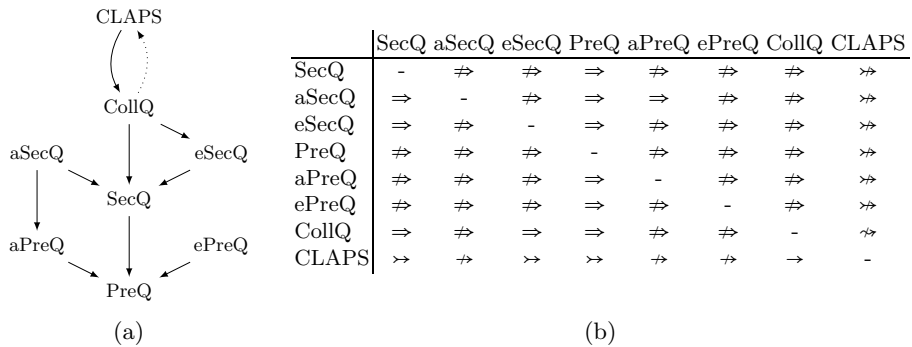

In this section, we examine the relationships among the properties in
Sect.~\ref{sec:defs}. Fig.~\ref{properties} illustrates these
graphically. The relationships among the properties with classical
analogs carry over from the classical setting, based on the framework
from~\cite{Song14}. The following is a sufficient criterion for
``lifting'' a reduction from a classical to a quantum setting:

\begin{lemma}[Corrollary 4.6 from~\cite{Song14}]
Let $\reduc=(\Gext,\trans,\Gint)$ be a black-box reduction that holds
for \PPT machines, and suppose the following:
\begin{enumerate}
\item $\Gint(\adv)$ and $\Gext(\adv)$ are defined for \QPT $\adv$;
\item $|\Adv{\Gint}(\adv)-\Adv{\Gext}{}(\trans(\adv))| \leq \negl(n)$
      for all \QPT $\adv$;
\item when $\trans$ runs $\adv$, it runs it ``in a straight line until
      completion,'' i.e., as an honest challenger would; and
\item for all $\adv,\adv'$ with $\Adv{\Gint}{}(\adv)=\Adv{\Gext}{}(\adv')$,
      $\Adv{\Gext}{}(\trans(\adv))=\Adv{\Gext}{}(\trans(\adv'))$.
\end{enumerate}
Then $\reduc$ holds for \QPT machines as well.
\label{crsl-reduction}
\end{lemma}

All the classical implication proofs from~\cite{RS04} ($\cimp$ from
Fig.~\ref{properties-table}) satisfy the hypotheses in
Lemma~\ref{crsl-reduction}, and thus that these proofs can be lifted
into a quantum setting. For example, the standard proof that
$\coll\imp\esec$ involves creating a reduction
$(\Gext=\coll,\trans,\Gint=\esec)$ where $\trans$ is defined as
follows:

\begin{mdframed}[style=figstyle,innerleftmargin=10pt,innerrightmargin=10pt]
\centering{\text{$\trans \: A \mapsto A'$}}
\begin{enumerate}
\item Sample $x \from \msg$ and send it to the challenger.
\item Receive $k$ from $\ch$.
\item Run $A(1^n,k,x)$ to get $x'$ and send $(x,x')$ to the challenger.
\end{enumerate}
\end{mdframed}

Note that $\trans$ could be applied to a quantum $\adv$ for $\qesec$
as easily as a classical one for $\esec$, and the result,
$\trans(\adv)$ finds a collision in the $\qcoll$ game. This is
guaranteed due to the classical ``interface'' in the definitions from
Sect.~\ref{sec:defs}. Moreover, it runs $A$ as normal. So hypotheses 1
and 3 from Lemma~\ref{crsl-reduction} hold. Hypothesis 2 holds as
well, since the success probabilities of $\adv$ and $\trans(\adv)$ are
the same Hypothesis 4 captures the idea that the success probability
of $\trans(\adv)$ depends only on the success probability of $\adv$,
not some specific facet of its internal behavior. This is easily seen
to be the case here.

The classical separations from~\cite{RS04} ($\csep$ from
Fig.~\ref{properties-table}) can also be lifted in a similar fashion.
For example, the proof that $\coll$ does not imply $\asec$ runs as
follows: Suppose that $H$ is $\coll$. We define a new function $H'$
such that if $k\not=0$, $H'_k(x)=H_k(x)$, but $H_0(x)=0$.
There is a trivial attack for $\asec$ on $H'$: The adversary simply
chooses $k=0$ and outputs any $x'\not=x$ as a second preimage.
Finally, we show that $H'$ is still collision resistant using a simple
reduction. The first half of this proof (the attack) is clearly as
possible on a quantum computer as it is on a classical one. In fact,
the structure of the properties from Sect.~\ref{sec:defs} (excluding
$\claps$)---where the adversary is given classical input and must
produce classical output---guarantees this. Moreover, as with the
implication proofs, the reductions in the separation proofs satisfy
the hypotheses in Lemma~\ref{crsl-reduction}. So we conclude that
these separations hold in a quantum setting as well.

We additionally examine the relationships between collapsing and each
of the standard properties. Unruh shows in~\cite{Unruh16_comm} that
collapsing implies collision resistance, and this proof applies to
$\qcoll$ as well. This leads to the transitive implications from
$\claps$ in Fig.~\ref{properties-table}. We find that $\claps$ does
\emph{not} imply $\qapre$, $\qasec$, or $\qepre$. The proofs of these
separations are given in Appendix~\ref{sec:CLAPS-seps}.

\section{Quantum security preservations of iterated hash constructions}
\label{sec:md}

In this section, we consider whether several standard iterated hash
constructions, including one in the random oracle model (ROX),
preserve the quantum-safe properties from Sect.~\ref{sec:defs}. The
constructions we consider are the same as those considered
in~\cite{ANPS07}, and we find that they preserve (and fail to
preserve) the quantum analogs of the same properties
that~\cite{ANPS07} show they do classically. In the case of the
standard constructions, we omit explicit proofs, instead using the
lifting framework we introduced in Sect.~\ref{sec:rels}. The proofs
for ROX, meanwhile, are more subtle, since they must be adapted to the
quantum random oracle model. We give explicit proofs in the most
interesting of these cases.

Andreeva et al. discuss eleven standard iterated hash constructions,
proving exhaustively (with a few exceptions) which of the seven
classical properties from~\cite{RS04} they preserve. These proofs are
amenable to being ``lifted'' to a quantum setting by reasoning similar
to that in Sect.~\ref{sec:rels}: Each implication proof uses a reduction
that satisfies the hypotheses of Lemma~\ref{crsl-reduction}. Each
separation combines an attack, which is still possible in a quantum
setting given the nature of the games we consider, and a reduction,
which also satisfies the hypotheses.

In contrast, we cannot use Lemma~\ref{crsl-reduction} to lift the
proofs for the random-oracle model construction ROX. In particular,
the reductions used cannot claim to run $\adv$ identically to an
honest challenger, since they must simulate a pair of random oracles.
This violates Hypothesis 3 of the lemma. Although the same
results hold, the proofs must be explicitly adapted, which we do
below.

\subsection{ROX preserves all quantum properties}
\label{sec:rox}

\begin{definition}[ROX]\label{def:rox}
$\ROX{H}_k = \ROXp{H}_k \circ \roxpad$
\begin{alignat*}{3}
        & \roxpad \:    &&\{0,1\}^*       &&\to     (\{0,1\}^b)^* \nonumber
\\      &               &&x               &&\mapsto \trunc_b(x \| \RO_2(\bar{x},|x|,1) \| \RO_2(\bar{x},|x|,2) \| \dots)
\\      & && && \nonumber
\\      & \ROXp{H}_k \: &&(\{0,1\}^b)^*   &&\to     \{0,1\}^d \nonumber
\\      &               &&\Lambda         &&\mapsto \IV \nonumber
\\      &               &&x_1\|\dots\|x_i &&\mapsto H_k(x_i \| \ROXp{H}_k(x_1\|\dots\|x_{i-1}) \oplus \RO_1(\bar{x},k,\nu(i)))
\end{alignat*}
where $\Lambda$ is the empty string; $\nu(i)$ is the largest integer
such that $2^{\nu(i)}$ divides i; $\bar{x}$ is the first $n$
bits of $x$; $\IV\in\bit^d$ is a fixed string; and
\begin{itemize}
\item $H \: \bit^n \times \bit^m \to \{0,1\}^d$, where $b=m-d>0$ is the block
      size, and $d \geq 2b$;
\item $\RO_1\:\bit^{2n+\log L}\to\bit^d$ and $\RO_2\:\bit^{n+2\log L}\to\bit^{2n}$
      are random oracles, where $L$ is the maximum input size in blocks; and
\item $\trunc_b$ truncates its input to a multiple of $b$ bits;
\end{itemize}
We denote the block length of $x$ as $\ell(x)=\ceil{(|x|+2n)/b}$, the
number of padding blocks as $1\leq
q_2(x)\leq\ceil{\frac{b+2n-1}{2n}}$, and the total oracle queries as
$q(x)=\ell(x)+q_2(x)$.
\end{definition}

Andreeva et al.~\cite{ANPS07} describe an iterated hash called ROX
(Definition~\ref{def:rox}) that preserves all of the classical
properties discussed in~\cite{RS04}. In addition to a compression
function, ROX relies on two random oracles ($\RO_{1,2}$), although it
does not rely on this fact for all proofs. Specifically, ROX
preserves $\apre$, $\pre$, $\asec$, and $\sec$ in the random oracle
(RO) model, and $\coll$, $\epre$, and $\esec$ in the standard model.

We show that ROX also preserves the quantum analogs of these
properties. Andreeva et al.'s standard-model proofs carry over nearly
unchanged for $\qcoll$, $\qepre$, and $\qesec$ carry over nearly
unchanged, so we omit those proofs. We show that ROX preserves
$\qapre$, $\qpre$, $\qasec$, and $\qsec$, replacing the classical RO
model with the QRO model.

We begin by stating the existence of some constructions using ROX that
will be useful in our proofs. The full constructions are given in
Appendix~\ref{sec:ROX-lemmas}.

\begin{lemma}[Extracting collisions on $H_k$ from collisions on
        $\ROX{H}_k$]
Given $\hat{x},\hat{x}'\in\dom(\ROX{H}_k)$ with $\hat{x}\not=\hat{x}'$
and $\ROX{H}_k(\hat{x})=\ROX{H}_k(\hat{x}')$, we can extract
$x,x'\in\dom(H_k)$ with $x\not=x'$ and $H_k(x)=H_k(x')$ except with
probability $\frac{1}{2^n}$ using $\ell(\hat{x})+\ell(\hat{x}')$
applications of $H$ and $q(\hat{x})+q(\hat{x}')$ oracle queries.
\label{ROX-collision-extraction}
\end{lemma}

\begin{lemma}[Embedding inputs for $H_k$ into inputs for $\ROX{H}$]
\label{ROX-message-embedding}
Given an input $x$ for $H_k$ and an index $i$, we can create an input
$\hat{x}$ for $\ROX{H}_k$ such that the input to the $i$th application
of $H_k$ is $x$ using $i$ calls to $H$ and at most
$\ceil{\frac{b+2n-1}{2n}}+i$ oracle queries. Moreover, an adversary
$\adv$ making $q$ queries notices the change with probability at most
$O(q^2/2^n)$.
\end{lemma}

\begin{lemma}[Extracting a preimage under $H$ from a preimage under $\ROX{H}$]
\label{ROX-preimage-extraction}
Given a key $k$ and a message $\hat{x} \in \{0,1\}^*$ with
$\ROX{H}_k(\hat{x}) = y$, we can generate a message $x \in \msg$ with
$H_k(x) = y$ using $\ell(x)-1$ calls to $H$ and $q(x)$ oracle queries.
\end{lemma}

We are now ready to prove that ROX preserves the properties from
Sect.~\ref{sec:defs} in the QRO model. To conserve space, we only
summarize our proofs here, providing the full proofs in Appendix
\ref{sec:ROX-pres}.

\begin{theorem}[ROX preserves $\qapre$]\label{ROX-aPre-pres}
If $H$ is $(t',\veps')$-$\qapre$, then $\ROX{H}$ is
$(t,\veps)$-$\qapre$ with
\begin{equation*}
    t = t' - \poly(n)\tau_H;\,
        \veps = \veps' + O(q^2/{2^d})
\end{equation*}
\end{theorem}

\paragraph{Proof summary.} We use a preimage target $y$ for $H$ as a
preimage target for $\ROX{H}$. In the classical proof, $y$ is
correctly distributed because an adversary would have to guess
correctly some random points to query $\RO_{1,2}$. This argument fails
in the quantum setting. We instead use QRO programming to show that
$y$ appears correctly distributed to a quantum adversary.

\begin{theorem}[ROX preserves $\qpre$]\label{ROX-Pre-pres}
If $H$ is $(t',\veps')$-$\qpre$ then $\ROX{H}$ is $(t,\veps)$-$\qpre$,
where
\begin{equation*}
    t = t - \poly(n)\tau_H;\,
        \veps' = \veps + O(q^2/{2^d})
\end{equation*}
\end{theorem}

\paragraph{Proof summary.} An adversary for $\qpre$ on $\ROX{H}$ can
be run using a preimage target $y$ for $H$, since $y$ will appear to
be correctly distributed. The argument is the same as that in the
proof of Theorem~\ref{ROX-aPre-pres}, so we omit it here for brevity.

\begin{theorem}[ROX preserves $\qasec$]\label{ROX-aSec-pres}
If $H$ is $(t',\veps')$-$\qasec$ then $\ROX{H}$ is $(t,\veps)$-$\qasec$
with
\begin{equation*}
    t = t' - \poly(n)\tau_H;\,
        \veps = \poly(n)\veps'/(1 - 1/2^n)(1 - \poly(n)/2^n)
\end{equation*}
\end{theorem}

\paragraph{Proof summary} We embed a second-preimage target for $H$
into a second-preimage target for $\ROX{H}$ by adaptively programming
$\RO_{1,2}$. We argue that reprogramming the random oracles in this
way is imperceptible to the adversary.

\begin{theorem}[ROX preserves $\qsec$]\label{ROX-Sec-pres}
If $H$ is $(t'\veps')$-$\qsec$ then $\ROX{H}$ is $(t,\veps)$-$\qsec$,
where
\begin{equation*}
    t = t' - \poly(n)\tau_H;\,
        \veps = \poly(n)\veps'/(1-1/2^n)(1-\poly(n)/2^n)
\end{equation*}
\end{theorem}

\paragraph{Proof summary.} Similarly to Theorem~\ref{ROX-aSec-pres},
here we embed a second-preimage target for $H$ into one for $\ROX{H}$
by programming $\RO_{1,2}$. Since we do not need to program
adaptively, however, the programming is straightforward. The proof is
similar to that of Theorem~\ref{ROX-aSec-pres}, so we omit it here for
brevity.

\bibliographystyle{alpha}
\bibliography{PQCrypto_2019_paper_13}

\appendix

\section{$\claps$ separation proofs}
\label{sec:CLAPS-seps}

Here we show that $\claps$ does not imply $\qapre$, $\qasec$, or
$\qepre$, completing the diagram in Fig.~\ref{properties}.

\subsection{$\claps \sep \qapre$}
\label{sec:CLAPS-aPre-sep}
\begin{theorem}\label{CLAPS-aPre-sep}
If a collapsing function family exists, then there is a function
family $H$ that is $\ts$-$\claps$ with negligible $\veps$, but with
$\Adv{\qapre}{H} = 1$.
\end{theorem}

\paragraph{Proof summary:} Since the key $k$ in the collapsing game is
chosen uniformly at random, a collapsing function can have a constant
number of ``bad'' keys that, for example, result in a constant
function $H_k$. Finding a preimage for such a function is obviously
trivial. Given a collapsing function $F$, we exhibit a function $H$
that is collapsing, but for which $\Adv{\qapre}{H} = 1$.

\begin{proof}

Suppose $F : \{0,1\}^n \times \{0,1\}^m \to \{0,1\}^{d}$ is
$\ts$-$\claps$, and define $H$ as follows:
        \[H_k(x) =
                \begin{cases}
                        0^d    & \text{if } k = 0^n
                \\      F_k(x) & \text{o.w.}
                \end{cases}\]

\begin{lemma}\label{CLAPS-aPre-sep-lemma1}
$H$ is $(t,\veps')$-$\claps$, where $\veps'=\veps-\frac{1}{2^n}$.
\end{lemma}

The obvious adversary suffices to break $\qapre$ on $H$: $A$ picks
$k=0^n$, and $B$ outputs any $x$. Since we assume $\veps$ to be
negligible, $\veps'$ is negligible as well, and the theorem is
immediate from Lemma \ref{CLAPS-aPre-sep-lemma1}.

\end{proof}

\begin{proof}[Proof of Lemma \ref{CLAPS-aPre-sep-lemma1}]
Let $(A,B)$ be a correct \QPT adversary with $\Adv{\claps}{H}(A,B) =
\veps'$. We construct a correct \QPT adversary $(A',B')$ for $F$ as
follows:

\grayframe{Constructing $(A',B')$ from $(A,B)$}{
\begin{minipage}[t]{0.48\textwidth}
\noindent $A'(1^n,k)$:
\begin{enumerate}
\item If $k = 0^n$, FAIL. Otherwise\dots
\item Run $A(1^n,k)$ to get $S,X,y$ with $\Pr[H_k(\Meas(X))=y] = 1$.
\item Send the challenger $S,X,y$.
\end{enumerate}
\end{minipage}
\hfill
\begin{minipage}[t]{0.48\textwidth}
\noindent $B'(1^n,S,X)$:
\begin{enumerate}
\item Run $B(1^n,S,X)$ to get $b$.
\item Send $b$ to the challenger.
\end{enumerate}
\end{minipage}
}

We claim that $(A',B')$ is correct: If $A'$ does not fail, then $k
\not= 0^n$, and by the premise that $(A,B)$ is correct,
$\Pr[F_k(\Meas(X))=y] = \Pr[H_k(\Meas(X))=y | k\not=0] = 1$.

Suppose $A'$ doesn't fail, and the challenger for $(A',B')$ is running
$\Game_1$. Then $B$ receives $S,\Meas(X)$, and $(A,B)$ sees $\Game_1$.
On the other hand, if the challenger for $(A',B')$ is running
$\Game_2$, $X$ is unmeasured, and $(A,B)$ sees $\Game_2$.

Since the probability that $A'$ fails is $\Pr[k=0^n] = \frac{1}{2^n}$,
$\Adv{\claps}{F}(A',B')\geq\veps'-\frac{1}{2^n}$. Therefore
$\Adv{\claps}{H} \leq \Adv{\claps}{F} + \frac{1}{2^n} = \negl(n)$.

\end{proof}

\subsection{$\claps \sep \qasec$}
\label{sec:CLAPS-aSec-sep}
\begin{theorem}\label{CLAPS-aSec-sep}
If a collapsing function family exists, then there is a function
family $H$ that is $\ts$-$\claps$ with negligible $\veps$, but with
$\Adv{aSec}{H} = 1$.
\end{theorem}

\paragraph{Proof summary:} This proof uses the same $H$ as the proof
of Theorem~\ref{CLAPS-aPre-sep}. Given a collapsing function $F$, this
$H$ is also collapsing, but has a single ``bad'' key $\hat{k}$ s.t.
$H_{\hat{k}}(x) = 0^d$, $\forall x \in \{0,1\}^m$. If an adversary can
control the key and chooses $\hat{k}$, then whichever $x$ the
challenger picks, any other element of $\dom(H_k)$ is a second
preimage. Since the proof is identical to
Theorem~\ref{CLAPS-aPre-sep}, we omit it for brevity.

\subsection{$\claps \sep \qepre$}
\label{sec:CLAPS-ePre-sep}
\begin{theorem}\label{CLAPS-ePre-sep}
If a collapsing function family exists, then there is a function
family $H$ that is $\ts$-$\claps$ with negligible $\veps$, but
$\Adv{ePre}{H} = 1$.
\end{theorem}

\paragraph{Proof summary:} If the image $\hat{y}$ of some element $x$
in the domain of $H$ is fixed, regardless of the key $k$, then an
adversary can find a preimage of an element $y$ of his choice easily:
All he must do is choose $y=\hat{y}$, and whatever $k$ the adversary
picks, $x$ is a preimage. But if $x$ is the only preimage of $y$, this
property does not help in creating superpositions of preimages to use
in the collapsing game. So $H$ may still be collapsing.

\begin{proof}
Suppose $F : \{0,1\}^n \times \{0,1\}^m \to \{0,1\}^{d-1}$ is
$\ts$-$\claps$. We define a new function $H : \{0,1\}^n \times
\{0,1\}^m \to \{0,1\}^d$ as follows:
\begin{align*}
H_k(x) =
      \begin{cases}
              0^d       & \text{if } x = 0^m
      \\      1\|F_k(x) & \text{o.w.}
      \end{cases}
\end{align*}

\begin{lemma}\label{CLAPS-ePre-sep-lemma1}
$H$ is $\ts$-$\claps$.
\end{lemma}

The obvious adversary suffices to break $\qepre$ on $H$: $A$ picks
$y=0^d$ and $B$ outputs $0^m$. Thus, the theorem is immediate from
Lemma \ref{CLAPS-ePre-sep-lemma1}.

\end{proof}

\begin{proof}[Proof of Lemma \ref{CLAPS-ePre-sep-lemma1}]
Let $(A,B)$ be a correct \QPT adversary with $\Adv{CLAPS}{H}(A,B) =
\veps'$. We construct a \QPT adversary $(A',B')$ for $F$ as follows:

\grayframe{Constructing $(A',B')$ from $(A,B)$}{
\begin{minipage}[t]{0.48\textwidth}
\noindent $A'(1^n,k)$:
\begin{enumerate}
\item Run $A(1^n,k)$ to get $S,X,y$ with $\Pr[H_k(\Meas(X))=y] = 1$.
\item Measure $\Meas_F(X)$ to get $y'$.
\item Send $S,X,y'$ to the challenger.
\end{enumerate}
\end{minipage}
\hfill
\begin{minipage}[t]{0.48\textwidth}
\noindent $B'(1^n,S,X)$:
\begin{enumerate}
\item Run $B(S,X,y)$ to get $b$.
\item Send $b$ to the challenger.
\end{enumerate}
\end{minipage}
}

$(A',B')$ is correct by construction. We claim that
$\Adv{CLAPS}{F}(A',B') \geq \epsilon$. Consider the following cases:
\begin{enumerate}
\item Suppose $y=0^d$. Then by the premise that $(A,B)$ is correct,
      $X=\ket{0^m}$, since $0^m$ is the only preimage of $0^d$. In this case
      $\Game_1=\Game_2$, since $X=\Meas(X)$, so
      $|\Pr[b=1\:\Game_1]-\Pr[b=1\:\Game_2]|=0$ for both $(A,B)$ and $(A',B')$.
\item Suppose $y\not=0^d$, and the challenger for $(A,B)$ is running $\Game_1$.
      Then $B$ receives $\Meas(X)$, so $B$ sees $\Game_2$.
\item Suppose $y\not=0^d$, and the challenger for $(A,B)$ is running $\Game_2$.
      By the premise that $(A,B)$ is correct, when $A$ produces $X$,
      \[X = \sum_{x \in H^{-1}(y)} \alpha_x \ket{x}
          = \ket{1} (\sum_{x \in F^{-1}(y')} \alpha_x \ket{x})\]
      So the measurement at step 2 of $A'$ does not collapse $X$. Hence, $B$
      sees $\Game_2$.
\end{enumerate}
Therefore $\Adv{CLAPS}{F}(A',B') = \veps'$.
\end{proof}

\section{$\qapre \equiv \sqapre$}
\label{sec:aPre-aPreSQ-equiv}

\begin{theorem}\label{aPre-aPreSQ-equiv}
Equations \ref{adv-aPre} and \ref{adv-aPreSQ} are equivalent. I.e.,
$\Adv{\qapre}{H} = \Adv{\sqapre}{H}$
\end{theorem}

\begin{figure}
\centering

\subfloat[b][$\Game_1$]{
\resizebox{0.26\textwidth}{!}{
\Qcircuit @C=0.7em @R=1.0em {
        \lstick{\ket{0^n}} & \gate{A} & \ustick{K}\qw & \multigate{1}{U_{H(x')}} &           \qw & \qw              & \meter & \qw & k
\\      \lstick{\ket{0^d}} & \qw      &           \qw & \ghost{U_{H(x')}}        & \ustick{Y}\qw & \multigate{1}{B} & \meter & \qw & y
\\      \lstick{\ket{0^m}} & \qw      &           \qw & \qw                      &           \qw & \ghost{B}        & \meter & \qw & x
}
\label{aPreSQ-game1}
}}
\hspace{1em}\raisebox{-5mm}{\large $\equiv$}\hspace{1.3em}
\subfloat[b][$\Game_2$]{
\resizebox{0.26\textwidth}{!}{
\Qcircuit @C=0.7em @R=1.0em {
        \lstick{\ket{0^n}} & \gate{A} & \ustick{K}\qw & \multigate{1}{U_{H(x')}} & \meter        & \qw              & \qw    & \qw & k
\\      \lstick{\ket{0^d}} & \qw      &           \qw & \ghost{U_{H(x')}}        & \ustick{Y}\qw & \multigate{1}{B} & \meter & \qw & y
\\      \lstick{\ket{0^m}} & \qw      &           \qw & \qw                      & \qw           & \ghost{B}        & \meter & \qw & x
}
\label{aPreSQ-game2}
}}
\hspace{1em}\raisebox{-5mm}{\large $\equiv$}\hspace{1.3em}
\subfloat[b][$\Game_3$]{
\resizebox{0.26\textwidth}{!}{
\Qcircuit @C=0.7em @R=1.0em {
        \lstick{\ket{0^n}} & \gate{A} & \ustick{K}\qw & \meter  & \multigate{1}{U_{H(x')}} & \qw           & \qw              & \qw    & \qw & k
\\      \lstick{\ket{0^d}} & \qw      &           \qw & \qw     & \ghost{U_{H(x')}}        & \ustick{y}\qw & \multigate{1}{B} & \meter & \qw & y
\\      \lstick{\ket{0^m}} & \qw      &           \qw & \qw     & \qw                      & \qw           & \ghost{B}        & \meter & \qw & x
}
\label{aPreSQ-game3}
}}

\caption{
Three equivalent games showing that measuring the $K$ wire before the
end does not change the $\sqapre$ game---hence $\sqapre$ and $\qapre$
are equivalent. Figure \ref{aPreSQ-game1} is the same as $\sqapre$;
Figure \ref{aPreSQ-game3} is functionally equivalent to $\qapre$; and
Figure \ref{aPreSQ-game2} is intermediate between the two. The state
register $S$ from Equations \ref{adv-aPre} and \ref{adv-aPreSQ} is
omitted for clarity.
}
\label{aPreSQ-games}
\end{figure}
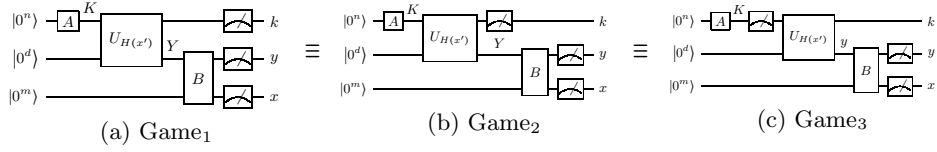

\begin{proof}[Proof of Theorem \ref{aPre-aPreSQ-equiv}]

Let each part $A$, $B$ of the adversary be a unitary operator and let
$n = \ceil{\lg |\ksp|}$ be the size of a key, $m = \ceil{\lg |\msg|}$
be the size of a message, and $d = \ceil{\lg |\dig|}$ be the size of a
digest. We illustrate the circuits for three games in Figure
\ref{aPreSQ-games}.

\begin{lemma}\label{aPre-aPreSQ-equiv-lemma1}
Game 1 is equivalent to Game 2.
\end{lemma}

\begin{lemma}\label{aPre-aPreSQ-equiv-lemma2}
Game 2 is equivalent to Game 3.
\end{lemma}

We claim that the theorem follows from Lemmas
\ref{aPre-aPreSQ-equiv-lemma1} and \ref{aPre-aPreSQ-equiv-lemma2}.
Note that Game 1 is exactly the $\sqapre$ game, as defined in Equation
\ref{adv-aPreSQ}. In Game 3, the output of $A$ and $B$, with the
exception of the state register $S$, are measured immediately, so
without loss of generality, we may assume that their output is
classical. Thus Game 3 is equivalent to the $\qapre$ game (Equation
\ref{adv-aPre}).

\end{proof}

In the following two proofs, let $\Meas_n$ denote measuring the first
$n$ qubits in the standard basis and leaving the rest untouched. In
other words,
\begin{equation}
\Meas_n(\rho) = \sum_{x\in\bit^n}(\kera{x} \otimes \Id)\rho(\kera{x} \otimes \Id)
\label{eqn:measurement}
\end{equation}

\begin{proof}[Proof of Lemma \ref{aPre-aPreSQ-equiv-lemma1}]

It suffices to show that for any unitary $B$ and all $x\in\bit^n$,
$(\kera{x} \otimes \Id)$ commutes with $(\Id_n \otimes B)$. This is
immediate, since $(\Id_n \otimes B)(\kera{x} \otimes \Id) = \kera{x}
\otimes B$.

\end{proof}

\begin{proof}[Proof of Lemma \ref{aPre-aPreSQ-equiv-lemma2}]

We must show that $U_{H(x')}$ commutes with $\Meas_n$. Since
\begin{align*}
U_{H(x')} \: \ket{K}\ket{Y} \mapsto \ket{K}\ket{H_k(x') \oplus Y}
\end{align*}
is a unitary operator that leaves the $K$ register untouched, it can
be viewed as using the $K$ register solely as control bits for $\CNOT$
gates, interspersed with unitaries on the $Y$ register. Let
$\CNOT_{i,j}$ denote a $\CNOT$ gate with control $i$ and target $j$,
and $V$ denote an arbitrary unitary that acts on the $Y$ register.
Then for $0 \leq i_\ell \leq n$ and $n+1 \leq j_\ell \leq n+d$ for all
$0 \leq \ell \leq q]$,

\begin{equation}
U_{H(x')} = (\Id^n \otimes V_0) \prod_{\ell=1}^q \CNOT_{i_\ell,j_\ell} (\Id^n \otimes V_\ell)
\label{eqn:controlled-unitary}
\end{equation}

Given the definition of $\Meas_n$ in Equation~\ref{eqn:measurement}, we must
show all the factors in Equation~\ref{eqn:controlled-unitary} commute with
$\kera{x} \otimes \Id_d$. Clearly this is the case for $\Id^n \otimes V_\ell$,
by the same reasoning as in the proof of Lemma~\ref{aPre-aPreSQ-equiv-lemma1}.
Similarly, it is well known that measurement of the control qubit commutes with
$\CNOT$.

\end{proof}

\section{ROX constructions}
\label{sec:ROX-lemmas}

\subsection{Extracting compression-function collisions
            (Lemma~\ref{ROX-collision-extraction})}
\begin{proof}[Proof of Lemma \ref{ROX-collision-extraction}]
We claim that the following procedure extracts a compression-function
with overwhelming probability:

\grayframe{\procedure{Extract-Collision}{k,\hat{x},\hat{x}'}}{
\begin{enumerate}
\item Let $\range{b}{1}{\ell}=\roxpad(\hat{x})$ and
      $\range{b'}{1}{\ell'}=\roxpad(\hat{x}')$.
\item Let $x_i=b_i\|\ROXp{H}_k(\range{b}{1}{i-1})$ and
      $x'_i=b'_i\|\ROXp{H}_k(\range{b'}{1}{i-1})$.
\item For $i=0$ to $\min(\ell,\ell')$ \dots
\item If $x_{\ell-i} \not= x'_{\ell'-i}$, output
      $(x=x_{\ell-i},x'=x'_{\ell'-i},i^*=\ell-i)$.
\end{enumerate}
}

Let $\bar{x},\bar{x}'$ and be the first $n$ bits of
$\hat{x},\hat{x}'$, and $\lambda,\lambda'$ be their respective
lengths. We must show that there always exists a colliding pair
$(x,x')$. We consider two cases:

\begin{itemize}

\item \emph{Case i.} Suppose $\bar{x}\not=\bar{x}'$ or
$\lambda\not=\lambda'$. Then since $x_\ell$ and $x_{\ell'}'$ each
contains at least one full output from $\RO_2$,
$x_\ell\not=x_{\ell'}'$ except with probability
$\delta=\frac{1}{2^n}$. In this case the inputs to the last
application of $H_k$ form a collision for $H_k$.

\item \emph{Case ii.} Otherwise, $\bar{x}=\bar{x}'$ and
$\lambda=\lambda'$. In this case, the padding applied to $\hat{x}$ and
$\hat{x}'$ will be identical. But the mask schedule taken from $\RO_1$
will be identical as well. So since $\hat{x}\not=\hat{x}'$, there must
be block pair, $(b_{i^*},b'_{i'^*})$ on which they differ. Since the
masks and padding match, these form a collision for $H_k$.

\end{itemize}

\end{proof}

\subsection{Embedding messages (Lemma~\ref{ROX-message-embedding})}
\begin{proof}[Proof of Lemma \ref{ROX-message-embedding}]
Let $h\|g = x$, where $|h| = b$ and $|g| = d$, and define the
following procedure:

\grayframe{\procedure{Embed-Message}{x,i}}{
\begin{enumerate}
\item Generate a random message $\hat{x}$ of length $\lambda \geq bi$.
\item If $i=1$, let $\bar{x}$ be the first $n$ bits of $x$, adding
      bits from $\hat{x}$, starting with the $(m+1)$st, if $x$ isn't
      long enough. Otherwise, let $\bar{x}$ be the first $n$ bits of
      $\hat{x}$.
\item Let $h_1\|\dots\|h_\ell = \roxpad(\hat{x})$ with $|h_j|=b$.
\item Evaluate $\mu_i = \ROXp{H}_k(x_1\|\dots\|x_{i-1})$.
\item Program $\RO_1(\bar{x},k,\nu(i))$ with $g \oplus \mu_i$.
      \label{step:programming1}
\item Program the $q'\leq q_2(\hat{x})$ outputs of $\RO_2$ contained
      in $h_i$ with $x$.
      \label{step:programming2}
\item Let $\hat{x}'$ be the first $\lambda$ bits of
      $h_1\|\dots\|h_{i-1}\|h\|h_{i+1}\|\dots\|h_\ell$.
\item Output $\hat{x}'$.
\end{enumerate}
}

In steps~\ref{step:programming1} and \ref{step:programming2}, the
above procedure requires us to program a random oracle. To do so, we
invoke witness search from~\cite{ES15}, where a witness is some image
of $\RO_{1,2}$ corresponding to an input that starts with $\bar{x}$.
Since $\bar{x}$ is chosen at random, and since the codomains of
$\RO_{1,2}$ are much larger than their domains, the random search
problem in~\cite{HRS16} can be reduced to this, with $2^n$ marked
items in a set of $2^{b-n}\geq2^n$, so the success probability is
$O(q^2/2^n)$.

\end{proof}

\subsection{Extracting compression-function preimages
            (Lemma~\ref{ROX-preimage-extraction})}
\begin{proof}[Proof of Lemma \ref{ROX-preimage-extraction}]
We claim that the following procedure extracts a collision-function
preimage with overwhelming probability:

\grayframe{\procedure{Extract-Preimage}{k,\hat{x}}}{
\begin{enumerate}
\item Evaluate $\ROX{H}_k(x)$ up to the last application of
$H_k$. Namely let
\begin{align*}
        & \range{x}{1}{\ell} = \roxpad(x)
\\      & x = x_\ell \| \ROXp{H}_k(\range{x}{1}{\ell-1})
                                \oplus \RO_1(\hat{x},k,\ell)
\end{align*}
\item Output $x$.
\end{enumerate}
}

By construction, $H_k(x)=y$ as desired. The only calls to $H$ and
$\RO_1,\RO_2$ are in the partial computation of $\ROX{H}(x)$. Since we
omit one call to $H$, the procedure calls it $\ell(x)-1$ times.

\end{proof}

\section{ROX property-preservation proofs}
\label{sec:ROX-pres}

\subsection{ROX preserves $\qapre$}
\begin{proof}[Proof of Theorem \ref{ROX-aPre-pres}]
Let $(A,B)$ be a $(t,\veps)$ quantum adversary for $\qapre$ on
$\ROX{H}$, making $q=\poly(n)$ oracle queries. We construct an
adversary $(A',B')$ for $H$ using an additional $\poly(n)$ oracle
queries: 
\grayframe{Constructing $(A',B')$ from $(A,B)$}{
\begin{minipage}[t]{0.48\textwidth}
\noindent $A'(1^n)$:
\begin{enumerate}
\item Run $k,S \from A^{\RO_{1,2}}(1^n)$, simulating quantum oracles $\RO_{1,2}$.
\item Output $k,S$.
\end{enumerate}
\end{minipage}
\hfill
\begin{minipage}[t]{0.48\textwidth}
\noindent $B'(1^n,S,y)$:
\begin{enumerate}
\item Run $x \from B^{\RO_{1,2}}(1^n,S,y)$.
\item Run \procedure{Extract-Preimage}{k,x} to obtain a preimage $x'$
      for $y$ under $H_k$.
\item Output $x'$.
\end{enumerate}
\end{minipage}
}

By Lemma \ref{ROX-preimage-extraction}, if $\ROX{H_k}(x)=y$, then
$H_k(x')=y$, and $A'$ wins the \qapre-game. Note that $y=H_k(g'\|h')$
for $g'\|h'$ chosen at random, while $A$ would expect a $\hat y$ to be
$H_k(g\|h)$, where $g$ contains at least $2b$ bits of
$\RO_2(\bar{x},\cdot)$, and $h=d\oplus\RO_1(\bar{x},\cdot)$ for some
$d$. The view of $(A,B)$ in the simulated run in $(A',B')$ is thus
identical to the real \qapre-game, unless $(A,B)$ can distinguish $y$
and $\hat y$ using at most $q$ queries. We show that $A$ can
distinguish them with probability at most $q^2/{2^n}$. Hence
$\Adv{\qapre}{H}(A',B') \geq \veps-q^2/2^n$.

We argue that if some challenger that knew $x$ were to reprogram
$\RO_{1,2}$ on inputs corresponding to $x$, no algorithm would be able
to discover this except with negligible probability. In the
Witness-Search game from \cite{ES15}, let $P(w)$ output 1 if and only
if $\RO_2(w,|x|,j)=g$ and
$\ROXp{H}_k(x_1\|\dots\|x_\ell)\oplus\RO_1(w,k,i)=h$ for some
$1\leq~i,j\leq|x|$. Next, let $w=\bar{x}$ and $pk=(k,|x|)$. This
amounts to finding a preimage with a suffix from a set in a random
function. Hence $\Adv{WS}{\Samp}(A) \leq O(q^2 /{2^d})$ by reducing a
random search problem developed in~\cite{HRS16} to it. Thus we can
safely reprogram $\RO_{1,2}$ at points corresponding to $P$ being
true, and $h,g$ are indistinguishable from the random values supplied
by $B'$.

\end{proof}

\subsection{ROX preserves $\qasec$}
\begin{proof}[Proof of Theorem \ref{ROX-aSec-pres}]
Let $(A,B)$ be a $(t,\veps)$ adversary for $\qasec$ on $\ROX{H}$
making $q=\poly(n)$ oracle queries. We construct an adversary
$(A',B')$ for $H$, using an additional
$\ceil{\frac{b+2n-1}{2n}}+\poly(n)$ oracle queries:

\grayframe{Constructing $(A',B')$ from $(A,B)$}{
\begin{minipage}[t]{0.48\textwidth}
\noindent $A'(1^n)$:
\begin{enumerate}
\item Run $k,S \from A^{\RO_{1,2}}(1^n)$, simulating quantum oracles
      $RO_{1,2}$.
\item Output $k,S$.
\end{enumerate}
\end{minipage}
\hfill
\begin{minipage}[t]{0.48\textwidth}
\noindent $B'(1^n,S,x)$:
\begin{enumerate}
\item Choose an index $i \leq \poly(n)$.
\item Run \procedure{Embed-Message}{x,i} to get
      $\hat{x}\in\dom(\ROX{H}_k)$ with $x$ embedded as the input to
      the $i$th application of $H_k$.
\item Run $B^{\RO_{1,2}}(1^n,S,\hat{x}')$, to get $\hat{x}'$.
\item Run \procedure{Extract-Collision}{k,\hat{x},\hat{x}'} to get
      $(x,x',i^*)$.
\item If $i^*\not=i$, FAIL. Output $x'$.
\end{enumerate}
\end{minipage}
}

By Lemma~\ref{ROX-message-embedding}, \textsc{Embed-Message} adds an
additional $i=\poly(n)$ applications of $H$ and an additional
$\ceil{\frac{b+2n-1}{2n}}+i$ oracle queries and alters the success
probability of $A$ by at most $O(q/2^n)=\poly(n)/2^n$, where
$q=\poly(n)$ is the number of queries $A$ makes. By
Lemma~\ref{ROX-collision-extraction}, \textsc{Extract-Collision} adds
$\ell(\hat{x})+\ell(\hat{x}')=\poly(n)$ applications of $H$ and
$q(\hat{x})+q(\hat{x}')=\poly(n)$ oracle queries and fails w.p.
$\frac{1}{2^n}$. Assuming both succeed, $i=i^*$ w.p.
$\frac{1}{\poly(n)}$. Hence $\Adv{\qasec}{H}(A',B') \geq
\veps(1-q/2^n)(1-1/2^n)/\poly(n)$.

\end{proof}

\end{document}